\documentclass[12pt]{article}

\usepackage{header}

\title{\Large{ {\bf Obviously Strategy-proof Choice of Social Acts}}}

\author{Abinash Panda\thanks{Department of Economics, Shiv Nadar Institution of Eminence, NH - 91, Gautam Buddha Nagar, Uttar Pradesh, India - 201314. Email: \texttt{ap280@snu.edu.in}} \  and Anup Pramanik\thanks{Department of Economics, Shiv Nadar Institution of Eminence, NH - 91, Gautam Buddha Nagar, Uttar Pradesh, India - 201314. Email: \texttt{anup.pramanik@snu.edu.in}.}}

\begin{document}
\maketitle
\begin{abstract}
We study obviously strategy-proof implementation in the framework of social choice over acts introduced by \cite{bahel2020strategyproof}. We characterize the class of unanimous social choice functions that are implementable via obviously strategy-proof mechanisms. Our main result shows that a unanimous social choice function is obviously strategy-proof implementable if and only if it is dictatorial. 
\end{abstract}

\noindent {\sc Keywords}: social choice functions; social acts;  obvious strategy-proofness; dictatorship\\

\noindent {\sc JEL Classification}: D71; D81; D82

\newpage
\section{Introduction}
We consider the model of social choice over social acts introduced by \cite{bahel2020strategyproof}. The framework studies collective decision making under uncertainty, where a social planner must choose a contingent social plan before the uncertainty is resolved. Formally, a social act specifies a social outcome in every possible state of nature, and agents evaluate such acts according to the subjective expected utility criterion. This environment provides a natural model for collective decisions in which uncertainty is subjective and agents may disagree not only about the desirability of outcomes, but also about the likelihood of states of nature. Examples include policy decisions, investment choices, and other collective choices made under uncertainty. 
Within this framework, \cite{bahel2020strategyproof} characterize the class of unanimous and strategy-proof social choice functions (SCFs). Their result shows that strategy-proofness and unanimity admit a rich class of rules, namely the locally bilateral top-selection SCFs. Thus, unlike many classical social choice environments where incentive compatibility leads to severe restrictions, the framework of social choice over social acts allows substantial flexibility under strategy-proofness.

Although strategy-proofness guarantees that truthful reporting is optimal, agents may nevertheless find it difficult to recognize that truthful behavior is indeed optimal, since doing so may require contingent reasoning. One formalization of what it means for participants to ``understand'' that a mechanism is strategy-proof is the concept of obvious strategy-proofness due to \cite{li2017obviously}. A mechanism (an extensive form game) is obviously strategy-proof (OSP) if every agent has an obviously dominant strategy. Informally, even the worst possible outcome from following this strategy is at least as good as the best possible outcome from any deviation. An SCF is OSP implementable if there exists an OSP mechanism that implements it.

In this paper, we consider the class of unanimous and strategy-proof SCFs characterized by \cite{bahel2020strategyproof} and ask which of these SCFs are OSP implementable. Our main result shows that the answer is highly restrictive. Although strategy-proofness and unanimity admit a rich class of SCFs in the framework of social choice over social acts, imposing obvious strategy-proofness leaves only dictatorial rules. Thus, the flexibility permitted by strategy-proofness collapses sharply under obvious strategy-proofness.

Our paper contributes to the growing literature on obvious strategy-proofness. Existing studies indicate that the implications of obvious strategy-proofness depend sharply on the economic environment. In voting environments with single-peaked domains, \cite{bade2016gibbard} and \cite{arribillaga2020obvious} characterize onto and OSP implementable SCFs. Since the class of strategy-proof and onto social choice functions on this domain coincides with the class of generalized median voter schemes (\cite{moulin1980strategy}), their characterization identifies an additional condition, namely the increasing intersection property, that generalized median voter schemes must satisfy in order to be OSP implementable. In division problems with single-peaked preferences, \cite{arribillaga2023all} show that all sequential allotment rules, a large subfamily of strategy-proof and efficient rules, are OSP implementable. By contrast, in assignment and matching environments, OSP is often highly restrictive. In the house allocation model of \cite{shapley1974cores}, where agents are initially endowed with indivisible objects but may prefer the objects owned by others, \cite{li2017obviously} shows that the classic top trading cycles mechanism is not obviously strategy-proof. \cite{troyan2019obviously} studies a more general class of top trading cycles mechanisms in which each object is associated with a strict priority ordering over agents, and characterizes OSP implementability through a condition on the priority structure called weak acyclicity.\footnote{See \cite{mandal2022obviously} for its generalization.} \cite{ashlagi2018stable} show that the deferred acceptance algorithm in the marriage model is not obviously strategy-proof for agents on the proposing side. They show, however, that the deferred acceptance algorithm becomes obviously strategy-proof for the proposing side on the restricted domain of acyclic preferences introduced by \cite{ergin2002efficient}. Our result contributes to this literature within the framework of social choice over social acts introduced by \cite{bahel2020strategyproof}. We show that OSP implementability is highly restrictive in this environment. In particular, within the class of unanimous and strategy-proof SCFs characterized by \cite{bahel2020strategyproof}, only dictatorial SCFs admit OSP implementations.

The remainder of the paper is organized as follows. Section \ref{TM} formally introduces the model. Section \ref{OSP} presents the notion of obvious strategy-proofness. Section \ref{TMR} contains the main result, and Section \ref{CN} concludes.

\section{The Model}\label{TM}
We follow the framework of \cite{bahel2020strategyproof}. There is a finite set of agents $N=\{1,\dots,n\}$ with $n\geq 2$. Uncertainty is represented by a finite set of states of nature $\Omega$, with $|\Omega|\geq 2$. Any subset $\Omega' \subseteq \Omega$ is called an event. The set of social outcomes is denoted by $X$, where $|X|\geq 3$. A (social) act is a function $f:\Omega \to X$, and we denote by $X^{\Omega}$ the set of all acts.

Agents have preferences over acts. Preferences are assumed to satisfy subjective expected utility (SEU) criterion. Thus, for each agent $i \in N$, there exist a valuation function 
$v_i:X \to \mathbb{R}$ and a probability measure 
$p_i:2^{\Omega} \to [0,1]$ such that for all $f,g \in X^{\Omega}$,
\[
f \succeq_i g \iff E_{p_i}^{v_i}(f) \ge E_{p_i}^{v_i}(g),
\]
where $E_{p_i}^{v_i}(h)= \sum_{\omega \in \Omega} p_i(\omega)\, v_i(h(\omega))$ for any act $h$.

A preference of agent $i$ is identified with its SEU representation $(v_i,p_i)$. Preferences are assumed to be strict, which implies that both the valuation function $v_i$ and the probability measure $p_i$ are injective. Without loss of generality, valuations are normalized so that $\max_{x \in X} v_i(x)=1$ and $\min_{x \in X} v_i(x)=0$. Let $\tau(v_i)$ denote the unique maximizer of $v_i$.

Let $\mathcal{V}$ denote the set of normalized injective valuation functions and $\mathcal{P}$ the set of injective probability measures. The domain of strict SEU preferences is
\[
\mathcal{D} = \left\{ (v_i,p_i) \in \mathcal{V} \times \mathcal{P} \;:\; E_{p_i}^{v_i}(f) \neq E_{p_i}^{v_i}(g) \text{ for all } f \neq g \right\}.
\]
A preference profile is denoted by $(v,p)=((v_1,p_1),\dots,(v_n,p_n))$, and $(v_{-i},p_{-i})$ denotes the profile of all agents except $i$.

\begin{defn}
A social choice function (SCF) is a mapping $\varphi:\mathcal{D}^n \to X^{\Omega}$.
\end{defn}

\begin{defn}
An SCF $\varphi$ is \emph{strategy-proof} if for all $i \in N$, all $(v,p) \in \mathcal{D}^n$, and all $(v_i',p_i') \in \mathcal{D}$,
\[
E_{p_i}^{v_i}(\varphi(v,p)) \ge E_{p_i}^{v_i}\big(\varphi((v_i',p_i'),(v_{-i},p_{-i}))\big).
\]
\end{defn}

\begin{defn}
An SCF $\varphi$ is \emph{unanimous} if for all $(v,p) \in \mathcal{D}^n$ and all $f \in X^{\Omega}$,
\[
\left[ E_{p_i}^{v_i}(f) \ge E_{p_i}^{v_i}(g) \ \text{for all } g \in X^{\Omega} \text{ and all } i \in N \right]
\implies \varphi(v,p)=f.
\]
\end{defn}

\cite{bahel2020strategyproof} characterize the class of strategy-proof and unanimous SCFs on the domain $\mathcal{D}$. The relevant class of rules is described next.

A probability measure $p_i$ is called a \emph{belief} if it is injective. Define an assignment (of states to agents) to be an $n$-component partition $\mathcal{A}=(A_1,\ldots,A_n)$ of the set of states $\Omega$, and denote the set of all assignments by $\mathcal{S}$. An assignment rule is a mapping $s:\mathcal{P}^n \to \mathcal{S}$, where $s(p)=(s_1(p),\ldots,s_n(p))$. Here, $s_i(p)$ is the event assigned to agent $i$ at the belief profile $p$. For any preference profile $(v,p)\in \mathcal{D}^n$ and any $\omega \in \Omega$, let $\varphi(v,p;\omega)$ denote the outcome of the act $\varphi(v,p)$ at state $\omega$.

\begin{defn}\rm 
An SCF $\varphi$ is called top-selection if there exists a unique assignment rule $s:\mathcal{P}^n\rightarrow \mathcal{S}$ such that for all $(v,p)\in \mathcal{D}^n$, for all $\omega \in \Omega$, and for all $i\in N$, we have   $$\omega \in s_i(p) \implies \varphi(v,p;\omega )=\tau(v_i).$$ 
In this case, the assignment rule $s$ is said to be associated with $\varphi$.
\end{defn}

For any nonempty event $\Omega' \subseteq \Omega$, define a belief on $\Omega'$ to be an injective probability measure on $2^{\Omega'}$, and denote the set of all such beliefs by $\mathcal{P}(\Omega')$. An assignment of $\Omega'$ is an $n$-component partition of $\Omega'$, and let $\mathcal{S}(\Omega')$ denote the set of all such assignments. An $\Omega'$-assignment rule is a mapping $s(\Omega'):\mathcal{P}(\Omega')^n \to \mathcal{S}(\Omega')$.

\begin{defn}
An $\Omega'$-assignment rule $s(\Omega'):\mathcal{P}(\Omega')^n \to \mathcal{S}(\Omega')$ is a \emph{constant rule} if there exists an $n$-component partition $\mathcal{A}$ of $\Omega'$ such that for all $p \in \mathcal{P}(\Omega')^n$, $s(\Omega')(p)=\mathcal{A}$.
\end{defn}

\begin{defn}
An $\Omega'$-assignment rule $s(\Omega'):\mathcal{P}(\Omega')^n \to \mathcal{S}(\Omega')$ is \emph{$(i,j)$-dictatorial} if there exists a proper covering\footnote{A proper covering of an event $\Omega'\subseteq \Omega$ is a collection $\{\Omega_k'\}_{k\in I}$ such that (i) $\bigcup_{k\in I} \Omega_k'=\Omega'$, (ii) $\Omega_k' \not\subset \Omega_\ell'$ for all $k \neq \ell$, and (iii) $\bigcap_{k\in I} \Omega_k'=\emptyset$.} $\mathbb{A}$ of $\Omega'$ such that for all $p \in \mathcal{P}(\Omega')^n$,
\[
\big(s_i(\Omega')(p),\, s_j(\Omega')(p)\big)
=
\big(\arg\max_{A \in \mathbb{A}} p_i(A),\, \Omega' \setminus \arg\max_{A \in \mathbb{A}} p_i(A)\big).
\]
An $\Omega'$-assignment rule is \emph{bilaterally dictatorial} if it is $(i,j)$-dictatorial for some (unique) ordered pair $(i,j)$.
\end{defn}

\begin{defn}
An $\Omega'$-assignment rule $s(\Omega'):\mathcal{P}(\Omega')^n \to \mathcal{S}(\Omega')$ is \emph{$(i,j)$-consensual} (with default $A$ such that 
$A\subset \Omega'$ and $A\neq \emptyset$) if for all $p\in \mathcal{P}(\Omega')^n$,
\[
\big(s_i(\Omega')(p),s_j(\Omega')(p)\big)
=
\begin{cases}
(\Omega'\setminus A,A), 
& \text{if } p_i(\Omega'\setminus A)>p_i(A)
\text{ and } p_j(A)>p_j(\Omega'\setminus A),\\
(A,\Omega'\setminus A), 
& \text{otherwise}.
\end{cases}
\]
An $\Omega'$-assignment rule is \emph{bilaterally consensual} if it is $(i,j)$-consensual for some pair of agents $(i,j)$.
\end{defn}

For any $p_i \in \mathcal{P}$ and any nonempty event $\Omega' \subseteq \Omega$, let $p_i|\Omega'$ denote the conditional belief of $p_i$ on $\Omega'$, defined by
\[
(p_i|\Omega')(A)=\frac{p_i(A)}{p_i(\Omega')} \quad \text{for all } A \subseteq \Omega'.
\]
Since $p_i$ is injective, $p_i|\Omega'$ is also injective.

\begin{defn}
An assignment rule $s:\mathcal{P}^n \to \mathcal{S}$ is \emph{locally bilateral} if there is a partition $\{\Omega^t\}_{t=1}^T$ of $\Omega$ and, for each $t=1,\ldots,T$, a constant, bilaterally dictatorial, or bilaterally consensual $\Omega^t$-assignment rule $s^t:\mathcal{P}(\Omega^t)^n \to \mathcal{S}(\Omega^t)$ such that
\[
s_i(p)=\bigcup_{t=1}^T s_i^t(p|\Omega^t)
\quad \text{for all } p \in \mathcal{P}^n \text{ and all } i \in N.
\]
\end{defn}

\begin{defn}
An SCF is a \emph{locally bilateral top selection} if it is a top-selection SCF whose associated assignment rule is locally bilateral.
\end{defn}

We recall the main characterization result of \cite{bahel2020strategyproof}.

\bigskip

\noindent\textbf{Theorem (Bahel and Sprumont (2020)).}
An SCF $\varphi$ is unanimous and strategy-proof if and only if it is a locally bilateral top selection.

\begin{defn}
An SCF $\varphi$ is \emph{dictatorial} if there exists an agent $i \in N$ such that for all $(v,p)\in \mathcal{D}^n$ and all $\omega \in \Omega$,
\[
\varphi(v,p;\omega)=\tau(v_i).
\]
The agent $i$ is called the \emph{dictator} in $\varphi$.
\end{defn}

Dictatorial SCFs belong to the class of locally bilateral top selections. In particular, the associated assignment rule assigns the entire state space to the dictator at every belief profile.

Our main result shows that every unanimous and obviously strategy-proof implementable SCF is dictatorial. We introduce the notion of obvious strategy-proofness in the next section.

\section{Obvious Strategy-Proofness}\label{OSP}
Following the notation of \cite{troyan2019obviously}, we now define obviously strategy-proof (OSP) implementable SCFs (\cite{li2017obviously}). We begin by describing mechanisms as extensive form games, which are used to define OSP implementability

\begin{defn}\label{Mechanism}\rm A mechanism is a tuple $G= <R,h,\iota,\alpha>$, where
\begin{itemize}
\item[1.] $R$ is a rooted tree where  
		
\begin{itemize}
\item[1.1] $r$ denotes the root node of $R$
\item[1.2] $L(R)$ denotes the set of leaves (end nodes) of $R$
\item[1.3] $V(R)$ denotes the set of internal nodes (vertices) of $R$
\item[1.4] For $\rm v\in V(R)$, $E(\rm v)$ denotes the set of outgoing edges from $\rm v$ 
\item[1.5] $\bar{E}(R)=\cup_{\rm v\in V(R)}E(\rm v)$ denotes the set of edges of $R$
\item[1.6] For $e\in \bar{E}(R)$, $\rho(e)\in V(R)$ denotes the source node of $e$
\end{itemize}
\item[2.] A function $h:L(R)\rightarrow X^{\Omega}$ from the leaves of $R$ to social acts
		 
\item[3.] A function $\iota:V(R)\rightarrow N$ describing which agent is to take an action at each internal node
		
\item[4.] A function $\alpha:\bar{E}(R)\rightarrow 2^{\mathcal{D}}\setminus\{\emptyset\}$ such that:
\begin{itemize}
\item[4.1] For all $e,e'\in \bar{E}(R)$ such that $\rho(e)=\rho(e')$, $\alpha(e)\cap \alpha(e')=\emptyset$
\item[4.2] For any $\rm v\in V(R)$, $\cup_{e\in E(\rm v)}\alpha(e)=\alpha(e')$, where $e'$ is the most recent edge along a path from $r$ to $\rm v$ such that $\iota(\rho(e'))=\iota(\rm v)$, or, if no such edge exists then $\cup_{e\in E(\rm v)}\alpha(e)=\mathcal{D}$ 
\end{itemize}
\end{itemize}
\end{defn}
	
Let $G$ be a mechanism as defined in Definition \ref{Mechanism}. For any preference profile $(v,p)\in \mathcal{D}^n$, the mechanism induces a unique path from the root $r$ to a terminal node. At each internal node, the assigned agent $i\in N$ chooses the edge that corresponds to their reported preference $(v_i,p_i)$. This process continues until reaching a terminal node, which uniquely determines the resulting social act $f\in X^{\Omega}$. 
	 
For any mechanism $G$, we write $\varphi^G$ to denote the SCF  implemented by $G$, where $\varphi^G(v,p)\in X^{\Omega}$ is the social act found via the above process. 
	
Two preferences $(v_i,p_i),(v_i',p_i')\in \mathcal{D}$ for agent $i\in N$ diverge at a node $\rm v\in V(R)$ if there exists two distinct edges $e,e'\in E(\rm v)$ outgoing from $\rm v$ such that $(v_i,p_i)\in \alpha(e)$ and $(v_i',p_i')\in \alpha(e')$. Next, we state the definition of OSP mechanism.

\begin{defn}\rm 
A mechanism $G=<R,\iota,h,\alpha>$ is OSP if for all $i\in N$, all nodes $\rm v\in V(R)$ such that $\iota(\rm v)=i$, and every $(v,p),(v',p')\in \mathcal{D}^n$ such that $(v_i,p_i)$ and $(v_i',p_i')$ diverge at $\rm v$, we have $E^{v_i}_{p_i}(\varphi^G(v,p))\geq E^{v_i}_{p_i}(\varphi^G(v',p'))$. 
\end{defn}
	
\begin{defn}\rm 
An SCF $\varphi$ is OSP implementable if there exists an OSP mechanism $G$ such that $\varphi=\varphi^G$.
\end{defn}

Note that every OSP implementable SCF is strategy-proof.

\section{The Main Result}\label{TMR}
In this section, we present the main result of the paper, which characterizes unanimous and obviously strategy-proof implementable social choice functions on the domain $\mathcal{D}$. We begin with a technical lemma that will play a key role in the proof of the theorem.

\begin{lemma}\rm \label{L1}
Consider any belief $p_i\in \mathcal{P}$. Suppose $X=\{x_1,\ldots,x_m\}$. Consider $v_i\in \mathcal{V}$ such that $v_i(x_1)=1,v_i(x_2)=t,\ldots,v_i(x_{m-1})=t^{m-2},v_i(x_m)=0$, where $t\in (0,1)$. Then, for all but finitely many $t\in (0,1)$, $(v_i,p_i)\in \mathcal{D}$. 
\end{lemma}

\begin{proof} Fix any two distinct social acts $f,g\in X^{\Omega}$, $f\ne g$. Consider the difference in expected utilities: 
$$\Delta (t)=E^{v_i}_{p_i}(f)- E^{v_i}_{p_i}(g)=\sum_{l=0}^{m-2}c_lt^l,$$ 
where  
$$c_l=\sum_{\{\omega\mid f(\omega)=x_{l+1}\}}p_i(\omega)\;-\sum_{\{\omega\mid g(\omega)=x_{l+1}\}}p_i(\omega).$$ 
First, we claim that $c_l\ne 0$ for some $l\in \{0,\ldots,m-2\}$. Since $f\neq g$, there exists $\omega^*\in \Omega$ such that
\[
f(\omega^*)\neq g(\omega^*).
\]
If $f(\omega^*)=x_m$, then, since $f(\omega^*)\neq g(\omega^*)$, we must have
\[
g(\omega^*)=x_k
\]
for some $k\in \{1,\ldots,m-1\}$.
Interchanging the roles of $f$ and $g$ if necessary, we may therefore assume that
\[
f(\omega^*)=x_k
\]
for some $k\in \{1,\ldots,m-1\}$. Define, $\Omega'=\{\omega'\in \Omega\mid f(\omega')=x_k\}$, $\Omega''=\{\omega'' \in \Omega \mid g(\omega'')=x_k\}$. By construction, $\Omega'\ne \Omega''$ as $\omega^*\in \Omega'$ and $\omega^*\notin \Omega''$. Since $p_i$ is injective,  $p_i(\Omega')\ne p_i(\Omega'')$. Hence,  $c_{k-1}\ne 0$. Thus, $\Delta(t)$ is a non-trivial polynomial of at most $m-2$ degree.
        
A non-trivial polynomial of finite degree has at most finitely many real roots. Thus, $\Delta(t)\ne 0$ for all but finitely many $t\in (0,1)$, meaning $f$ and $g$ are not indifferent for all but finitely many values of $t$.

Since $X^{\Omega}$ is finite, there are only finitely many pairs $(f,g)$ with $f\neq g$. For each such pair, the set of values $t\in (0,1)$ for which
\[
E_{p_i}^{v_i}(f)=E_{p_i}^{v_i}(g)
\]
is finite. Taking the union over all distinct pairs of acts, the set of values $t\in (0,1)$ for which some pair of distinct acts is indifferent remains finite. Hence, for all but finitely many $t\in (0,1)$, every pair of distinct acts is strictly ranked under $(v_i,p_i)$.
\end{proof}

We now state the main result of the paper.

\begin{theorem}\rm 
An SCF $\varphi$ is unanimous and OSP implementable if and only if it is dictatorial.
\end{theorem}
	
\begin{proof} Let $\varphi$ be a unanimous and OSP implementable SCF. Since every OSP implementable SCF is strategy-proof, it follows from \cite{bahel2020strategyproof} that $\varphi$ is a locally bilateral top-selection SCF. Let $s:\mathcal{P}^n\to \mathcal{S}$ denote the associated locally bilateral assignment rule. That is, there exists a partition $\{\Omega^t\}_{t=1}^{T}$ of $\Omega$ and, for each $t\in \{1,\ldots,T\}$, an $\Omega^t$-assignment rule
\[
s^t:\mathcal{P}(\Omega^t)^n \rightarrow \mathcal{S}(\Omega^t)
\]
which is either constant, bilaterally consensual, or bilaterally dictatorial, such that
\[
s_i(p)=\bigcup_{t=1}^T s_i^t(p|\Omega^t)
\]
for all $p\in \mathcal{P}^n$ and all $i\in N$.

We complete the proof through a sequence of claims. 

\begin{claim}\rm \label{C1}
For every agent $i\in N$, either $s_i(p)\neq \emptyset$ for all $p\in \mathcal{P}^n$, or $s_i(p)=\emptyset$ for all $p\in \mathcal{P}^n$.
\end{claim}
\begin{proof}
Fix an agent $i\in N$. Suppose that $s_i(p)\neq \emptyset$ for some $p\in \mathcal{P}^n$. We show that $s_i(p')\neq \emptyset$ for every $p'\in \mathcal{P}^n$.

Since $s_i(p)\neq \emptyset$, there exists $\omega\in s_i(p)$. As $\{\Omega^t\}_{t=1}^{T}$ is a partition of $\Omega$, there exists $t\in \{1,\ldots,T\}$ such that $\omega\in \Omega^t$. Since
\[
s_i(p)=\bigcup_{r=1}^{T}s_i^r(p|\Omega^r),
\]
we have
\[
\omega\in s_i^t(p|\Omega^t).
\]
In particular,
\[
s_i^t(p|\Omega^t)\neq \emptyset.
\]

Now consider the $\Omega^t$-assignment rule $s^t$. By construction, $s^t$ is either constant, bilaterally consensual, or bilaterally dictatorial. In each of these cases, if agent $i$ receives a nonempty event at one belief profile, then agent $i$ receives a nonempty event at every belief profile. Hence,
\[
s_i^t(p'|\Omega^t)\neq \emptyset
\]
for every $p'\in \mathcal{P}^n$. Therefore,
\[
s_i(p')
=
\bigcup_{r=1}^{T}s_i^r(p'|\Omega^r)
\neq \emptyset
\]
for every $p'\in \mathcal{P}^n$.

Thus, if $s_i(p)\neq \emptyset$ for some $p\in \mathcal{P}^n$, then $s_i(p')\neq \emptyset$ for all $p'\in \mathcal{P}^n$. Equivalently, for every agent $i\in N$, either $s_i(p)\neq \emptyset$ for all $p\in \mathcal{P}^n$, or $s_i(p)=\emptyset$ for all $p\in \mathcal{P}^n$.
\end{proof}

\begin{claim}\rm \label{C2}
There exists an agent $i\in N$ such that $s_i(p)\neq \emptyset$ for all $p\in \mathcal{P}^n$.
\end{claim}

\begin{proof}
Fix any belief profile $p\in \mathcal{P}^n$ and any state $\omega\in \Omega$. Since $s(p)$ is an assignment of $\Omega$, there exists an agent $i\in N$ such that
\[
\omega\in s_i(p).
\]
Thus, $s_i(p)\neq \emptyset$. By Claim \ref{C1}, it follows that
\[
s_i(p')\neq \emptyset
\]
for all $p'\in \mathcal{P}^n$. Hence, there exists an agent $i\in N$ such that $s_i(p)\neq \emptyset$ for all $p\in \mathcal{P}^n$.
\end{proof}
	
\begin{claim}\rm \label{C3}
There exists a unique agent $i\in N$ such that $s_i(p)\neq \emptyset$ for all $p\in \mathcal{P}^n$.
\end{claim}
\begin{proof}
By Claim \ref{C2}, there exists at least one agent $i\in N$ such that
\[
s_i(p)\neq \emptyset
\]
for all $p\in \mathcal{P}^n$. We show that there cannot be more than one such agent.

Suppose, to the contrary, that there are two distinct agents with this property. Without loss of generality, let these agents be $1$ and $2$. Fix any belief profile $p\in \mathcal{P}^n$. Since $s_1(p)\neq \emptyset$ and $s_2(p)\neq \emptyset$, write
\[
A_1=s_1(p), \qquad A_2=s_2(p),
\]
and let
\[
A'=\Omega\setminus (A_1\cup A_2).
\]
Since $s(p)$ is an assignment, $A_1$ and $A_2$ are disjoint. Also, since $A_1,A_2\neq \emptyset$, injectivity of $p_1$ implies
\[
p_1(A_1)>0
\quad\text{and}\quad
p_1(A_2)>0.
\]

Choose three distinct outcomes $a,b,c\in X$. By Lemma \ref{L1}, choose preferences so that the following two profiles belong to $\mathcal{D}^n$. Consider first a profile $(v,p)$ such that $\tau(v_1)=a$, $\tau(v_2)=b$, and $\tau(v_k)=a$ for all $k\in N\setminus\{1,2\}$. Since $\varphi$ is a top-selection SCF, the outcome $f=\varphi(v,p)$ is given by
\[
f(\omega)=
\begin{cases}
a & \text{if } \omega\in A_1\cup A',\\
b & \text{if } \omega\in A_2.
\end{cases}
\]

Next, consider a second profile $(v',p)$ such that $\tau(v_1')=c$, $\tau(v_2')=a$ and $v_k'=v_k$ for all $k\in N\setminus\{1,2\}$. Then $g=\varphi(v',p)$ is given by
\[
g(\omega)=
\begin{cases}
c & \text{if } \omega\in A_1,\\
a & \text{if } \omega\in A_2\cup A'.
\end{cases}
\]
Clearly, $f\neq g$.

We now construct auxiliary valuation functions $v_1^*$ and $v_2^*$ such that $(v_1^*,p_1)$ and $(v_2^*,p_2)$ belong to $\mathcal{D}$ as follows. Let $v_1^*\in \mathcal{V}$ be such that
\[
v_1^*(a)=1,\qquad v_1^*(c)=t,\qquad v_1^*(b)=0,
\]
with the remaining outcomes assigned distinct values between $0$ and $t$ in the form required by Lemma \ref{L1}. For $t\in(0,1)$, we have
\[
E^{v_1^*}_{p_1}(f)=p_1(A_1)+p_1(A')
\]
and
\[
E^{v_1^*}_{p_1}(g)=t\,p_1(A_1)+p_1(A_2)+p_1(A').
\]
Since $p_1(A_2)>0$, we may choose $t\in(0,1)$ sufficiently close to $1$ so that
\[
E^{v_1^*}_{p_1}(f)<E^{v_1^*}_{p_1}(g).
\]
By Lemma \ref{L1}, we may choose such a $t$ so that $(v_1^*,p_1)\in \mathcal{D}$. Moreover, $\tau(v_1^*)=a=\tau(v_1)$.

Similarly, choose $v_2^*\in \mathcal{V}$ such that $\tau(v_2^*)=a=\tau(v_2')$ and
\[
E^{v_2^*}_{p_2}(g)<E^{v_2^*}_{p_2}(f),
\]
with $(v_2^*,p_2)\in \mathcal{D}$.

Since $\varphi$ is top-selection and only peaks matter for the selected acts, we have
\[
\varphi((v_1^*,p_1),(v_2,p_2),(v_{-\{1,2\}},p_{-\{1,2\}}))=f
\]
and
\[
\varphi((v_1',p_1),(v_2^*,p_2),(v'_{-\{1,2\}},p_{-\{1,2\}}))=g.
\]

Let
\[
\theta=((v_1^*,p_1),(v_2,p_2),(v_{-\{1,2\}},p_{-\{1,2\}}))
\]
and
\[
\theta'=((v_1',p_1),(v_2^*,p_2),(v'_{-\{1,2\}},p_{-\{1,2\}})).
\]

Since $\varphi$ is OSP implementable, there exists an OSP mechanism $G$ such that $\varphi=\varphi^G$. The mechanism $G$ assigns $f$ at $\theta$ and $g$ at $\theta'$. Since $f\neq g$, the two induced paths in $G$ must diverge. Let ${\rm v}$ be the node at which these two paths diverge. The profiles $\theta$ and $\theta'$ differ only in the preferences of agents $1$ and $2$. Hence the divergence node ${\rm v}$ cannot be assigned to any agent other than $1$ or $2$. Therefore, $\iota({\rm v})\in\{1,2\}$.

If $\iota({\rm v})=1$, then the two preferences $(v_1^*,p_1)$ and $(v_1',p_1)$ diverge at ${\rm v}$. Since $G$ is OSP, we must have
\[
E^{v_1^*}_{p_1}(\varphi^G(\theta))
\geq
E^{v_1^*}_{p_1}(\varphi^G(\theta')).
\]
That is,
\[
E^{v_1^*}_{p_1}(f)
\geq
E^{v_1^*}_{p_1}(g),
\]
contradicting the construction of $v_1^*$.

If $\iota({\rm v})=2$, then the two preferences $(v_2,p_2)$ and $(v_2^*,p_2)$ diverge at ${\rm v}$. Since $G$ is OSP, we must have
\[
E^{v_2^*}_{p_2}(\varphi^G(\theta'))
\geq
E^{v_2^*}_{p_2}(\varphi^G(\theta)).
\]
That is,
\[
E^{v_2^*}_{p_2}(g)
\geq
E^{v_2^*}_{p_2}(f),
\]
contradicting the construction of $v_2^*$.

Therefore, there exists at most one agent $i\in N$ such that
$s_i(p)\neq \emptyset$, for all $p\in \mathcal{P}^n$. Combining this with Claim \ref{C2}, it follows that there exists a unique such agent.
\end{proof}

By Claim \ref{C3}, there exists a unique agent $i\in N$ such that
$s_i(p)\neq \emptyset$ for all $p\in \mathcal{P}^n$. Then, by Claim \ref{C1}, for every agent $j\neq i$, $s_j(p)=\emptyset$
for all $p\in \mathcal{P}^n$. Since $s(p)$ is an assignment of $\Omega$ for every $p\in \mathcal{P}^n$, it follows that $s_i(p)=\Omega$ for all $p\in \mathcal{P}^n$. Hence, $\varphi$ is dictatorial. Conversely, it is straightforward to verify that every dictatorial SCF is unanimous and can be implemented by an OSP mechanism.
\end{proof}

\section{Conclusion}\label{CN}
We have characterized the class of unanimous social choice functions that admit OSP implementations in the framework of social choice over acts introduced by \cite{bahel2020strategyproof}. While \cite{bahel2020strategyproof} identify a rich class of unanimous and strategy-proof SCFs, our analysis shows that obvious strategy-proofness is far more restrictive in this environment. In particular, only dictatorial SCFs admit OSP implementations. This highlights the strong tension between incentive simplicity and non-dictatorial collective choice in environments with subjective expected utility preferences over social acts.

\bibliographystyle{ecta}
\bibliography{order}

\end{document}